\documentclass[[10pt,oneside,a4paper]{article}
\title{\textbf{Entanglement Measure Based on Optimal Entanglement Witness}}
\usepackage{authblk}
\usepackage[T1]{fontenc}
\usepackage[utf8]{inputenc}
\author{Nan Yang$^1$,\quad Jiaji Wu$^1$,\quad Xianyun Dong$^1$,\quad Longyu Xiao$^1$,\quad Jing Wang$^1$\thanks{wangjing3310@163.com},\quad Ming Li$^1$}
\affil{$^1$College of Science, China University of Petroleum, 266580 Qingdao, P.R. China}
\usepackage{amsthm}
\usepackage{geometry}
\usepackage{multicol}
\usepackage{graphicx}
\usepackage{subfigure}
\usepackage{indentfirst}
\usepackage{caption}
\usepackage{float}
\usepackage{caption}
\usepackage{ragged2e}
\theoremstyle{definition}
\newtheorem{Def}{Definition}
\newtheorem{Cor}{Corollary}
\newtheorem{Pro}{Propositon}
\newtheorem{mythm}{Theorem}
\begin{document}
	\maketitle
	\begin{abstract}
		We introduce a new entanglement measure based on optimal entanglement witness. First of all, we show that the entanglement measure satisfies some necessary properties, including zero entanglements for all separable states, convexity, continuity, invariance under local unitary operations and non-increase under local operations and classical communication(LOCC). More than that, we give a specific mathematical expression for the lower bound of this entanglement measure for any bipartite mixed states. We further improve the lower bound for 2$ \otimes $2 systems. Finally, we numerically simulate the lower bound of several types of specific quantum states.
	\end{abstract}
	
	\section{Introduction}
	
		Entanglement is one of the most remarkable features of quantum mechanics, which has been recognized as an essential resource in quantum information theory \cite {F. L. Yan2011,T. Gao2008}. In recent years, quantum entanglement is increasingly the focus of people's attention and is widely applied to quantum information processing tasks \cite {M.A. Nilsen2000}.
	
		Since entanglement is a newly discovered state of resources, it's of vital importance to discover the mathematical structure behind its theory. Entanglement measure is the quantization of entanglement and various entanglement measures have been well defined. Arindam Lala \cite {Arindam Lala2020} has studied various entanglement measures associated with certain non-conformal field theories. And Jacob L. Beckey et.al. \cite {Jacob L. Beckey2021} have proposed  entanglement measures that are computable and operationally meaningful multipartite.
	
		Entanglement witness(EW) gives a sufficient condition for detecting entanglement, which is equivalent to making a hyperplane in the quantum state space to separate some separable states from entangled states. In Ref. \cite {Cheng-Jie Zhang2007}, the authors have studied the optimal EW(OEW) based on local orthogonal observables. And the optimal entanglement witness constructed by Ref. \cite {Minh Tam2021} could effectively detect the entangled states produced by cooper pair splitters.
	
		In this paper, we introduce a new entanglement measure of bipartite quantum states, which is closely related to optimal entanglement witness. In Sec.III, we prove explicitly that this new measure satisfies many properties of bipartite entanglement measures. In Sec.IV, we obtain the lower bound of this measure for any bipartite mixed states. Furthermore, considering the property that there is a double cover relationship between SU(2) group and SO(3) group, we give a better lower bound on $ 2 \otimes 2 $ space. Finally, we present numerical simulations of the lower bound for some specific forms of quantum states.
	
	\section{Construction of Entanglement Measure}
	
		In Ref. \cite {Jinchuan Hou2010}, we know that EW could be the operator in the form of $ \alpha $$ I+L $. Let $ \mathcal {H}_{A} $ and  $ \mathcal {H}_{B} $ be arbitrary finite-dimensional Hilbert spaces. If an EW $ W_{E} $ acting on $\mathcal {H}_{A} $ and $ \mathcal {H}_{B} $ can be written in the following form:
	\begin{eqnarray}
		&W_{E}=\alpha (L) I-L,\label {eqa1}
	\end{eqnarray}
		where $ I $ is the identity operator, $ L $ is a self-adjoint operator, and $ \alpha(L) $ is the maximum of expectation values taken over all separable pure states $ | \varphi \rangle $:
	\begin{eqnarray}
		& \alpha = \max \limits_{| \varphi \rangle \in sep} \langle \varphi |L| \varphi \rangle,
	\end{eqnarray}
		then this EW is called an OEW. We use the results of this OEW to quantify entanglement.
	\begin{Def}
		Let $ M_{a} $ be a positive semi-definite set
		\begin{equation}
			M_{a}= \{ H \geq 0\ |\ H \in \mathcal {H}_{A} \otimes \mathcal{H}_{B}, tr(H)=a, a \in R^{+} \}. \label {e1}
		\end{equation}
	
		The entanglement measure $ W_{E} $ of a bipartite state $ \rho \in ( \mathcal{H}_{A} \otimes \mathcal {H}_{B}) $ is given by
		\begin{eqnarray} \label {md}
			&C_{w}( \rho ) = \max \limits _{L \in M_{a}} \{-tr(W_{E} \rho) \}.
		\end{eqnarray}
	\end{Def}
	
		We define $ ||M||_{KF}=Tr[ \sqrt {MM^{ \dagger }}] $, $ ||M||_{HS}= \sqrt {Tr[MM^{ \dagger }]} $ and $ ||M||_{2}= \sigma _{max}(M) $. They respectively represent the trace norm, Frobenius norm and 2-norm of the matrix $ M $. Meanwhile, we could also get the following result according to the relevant knowledge of EW.
	\begin{eqnarray}
		tr(W_{E} \rho ) < 0, \\
		tr(W_{E} \sigma ) \geq 0,\label {eqa6}
	\end{eqnarray}
		where $ \rho$ is an arbitrary entangled state acting on $ \mathcal{H}_{A} \otimes \mathcal {H}_{B} $ and $ \sigma $ is an arbitrary separable state acting on $ \mathcal{H}_{A} \otimes \mathcal {H}_{B} $. It's not difficult to see that $ W_{E} $ is one type of OEW, and we could get the above inequalities from the Ref. \cite {J. Sperling2009}.
	
	\section{Properties of Entanglement Measure}
	
		For $ \forall \rho \in \mathcal {H}_{A} \otimes \mathcal {H}_{B} $, an entanglement measure is a functional $ C $ defined on the set of density operators on the Hilbert space. A well-defined measure of entanglement $ C $ should satisfy the following requirements:
	
	\noindent (E1) $ C( \rho )=0 $, for any $ \rho \in Sep$(zero for all separable states).
	
	\noindent (E2) $ C( \lambda \rho _{1}+(1- \lambda ) \rho _{2}) $$ \leq \lambda C( \rho _{1})+(1- \lambda ) C( \rho _{2})$, for all $ \rho_{1} $, $ \rho_{2} $ and $ 0 \leq \lambda \leq 1 $ (convexity).
	
	\noindent (E3)  Keeping invariant under local unitary transformations, i.e.,
	
	\noindent $$ C((U_{1} \otimes U_{2})^{ \dagger } \rho(U_{1} \otimes U_{2}))= C( \rho) $$ holds for all unitary operators $U_{1}$ and $U_{2}$.
	
	\noindent (E4) $ C( \Lambda_{LOCC}(\rho))$$ \leq C( \rho) $(nonincreasing under local operations and classical communication (LOCC)).
	
	\noindent (E5) When $ \| \rho - \sigma \| \rightarrow 0$, $C( \rho)-C( \sigma) \rightarrow 0$(continuity).
	
		Subsequently, we will prove that our entanglement measure defined in Eq.(\ref{md}) fulfills these requirements.
	
	\begin{Pro}
		$C_{w}( \rho)$ satisfies property (E2).
	\end{Pro}
	
	\begin{proof}[ \textbf{Proof}]
		Let $ \rho = \lambda \rho_{1}+(1- \lambda ) \rho_{2} $, where $0 \leq \lambda \leq 1$. By substituting it into $C_{w}( \rho)$ we have
		\begin{eqnarray}
			C_{w}( \rho )&=&C_{w}( \lambda \rho_{1}+(1-\lambda)\rho_{2})\nonumber\\
			&=&\max\limits_L\{-\alpha(L)+\lambda tr(L\rho_{1})+(1-\lambda) tr(L\rho_{2})\}\nonumber\\
			&=&\max\limits_L\{\lambda (-\alpha(L)+tr(L\rho_{1}))+(1-\lambda) (-\alpha(L)+tr(L\rho_{2}))\}\nonumber\\
			&\leq& \lambda C_{w}(\rho_{1})+(1-\lambda) C_{w}(\rho_{2}).\label{4}
		\end{eqnarray}
	\end{proof}
	
	\begin{Pro}
		$C_{w}(\rho)$ satisfies property (E1).
	\end{Pro}
	
	\begin{proof}[\textbf{Proof}]
		Assume that $ \rho $ is a separable and pure state acting on the space of $ \mathcal{H}_{A}\otimes\mathcal{H}_{B} $, i.e. $\rho $ could be written as $ \rho = \rho_{1} \otimes \rho_{2}= | \varphi_{1} \rangle \langle \varphi_{1}| \otimes | \varphi_{2}\rangle \langle \varphi_{2}| $, where $ | \varphi_{1} \rangle \in \mathcal{H}_{A}$, $ | \varphi_{2} \rangle \in \mathcal{H}_{B} $. Then we have
		\begin{eqnarray}
			tr( W_{E} \rho)&=&tr( \alpha \rho -L \rho) \nonumber\\
			&=&tr( \alpha)-tr( L \rho _{1} \otimes \rho _{2}) \nonumber\\
			&=&tr(\max \langle \varphi^{A}| \otimes \langle \varphi^{B}| L | \varphi^{A} \rangle \otimes | \varphi^{B} \rangle )-tr( L\rho_{1} \otimes \rho_{2}) \nonumber\\
			&=& \max \limits_{ \rho^{A*},\rho^{B*}}tr(L \rho^{A*} \otimes \rho^{B*})-tr(L \rho_{1} \otimes \rho_{2}),
		\end{eqnarray}
		where $ \rho^{A*}= | \varphi^{A*} \rangle \langle \varphi^{A*}| $ and $ \rho^{B*}= | \varphi^{B*} \rangle \langle \varphi^{B*}| $. According to the definition and Eq.(\ref{eqa6}) we have $ -tr ( W_{E} \rho ) \leq 0 $. Let $ L= \frac{a}{n} I $, where $ n $ is the dimension of $ L $, we have $ W_{E}=0 $. Thus the entanglement degree for the separable pure states is 0, which illustrate that
		\begin{eqnarray}
			C_{w}( \rho ) = \max \limits_L\{ -tr(W_{E} \rho) \}=0.
		\end{eqnarray}
		
		The above-mentioned fact is discussed when $ \rho $ is a separable pure state acting on the space of $ \mathcal{H}_{A}\otimes\mathcal{H}_{B} $. If $ \rho $ is a separable mixed state, we have a decomposition of $ \rho $:
		\begin{eqnarray}
			\rho = \sum_{i}^{}p_{i}  | \varphi _{1i}  \rangle \langle \varphi _{1i}| \otimes| \varphi _{2i}  \rangle \langle \varphi _{2i}| 	.
		\end{eqnarray}
		
		According to property (E2), we know $ C_{w}(\rho) $ satisfies
		\begin{eqnarray}
			C_{w}  ( \rho   ) \leq \sum_{i}^{}p_{i}C (| \varphi _{1i}\rangle \langle \varphi _{i}| \otimes | \varphi _{2i}  \rangle \langle \varphi _{2i}| )=0.
		\end{eqnarray}
		
		Note that when $ L= \frac{a}{n}I $, we have $ C_{w}( \rho )=0 $. Thus for $ \forall \rho $ that is separable,  $C_{w}( \rho)=0 $.
	\end{proof}
	
	\begin{Pro}
		$ C_{w}( \rho) $ has local unitary transformation invariance, satisfying property(E3).
	\end{Pro}
	
	\begin{proof}[\textbf{Proof}]
		For $ U^{\dagger }\rho U $, we write the measure as:
		\begin{eqnarray}
			C_{w}(U^{ \dagger } \rho U)&=& \max \limits_{L}\{-tr(W_{E} U^{ \dagger } \rho U) \} \nonumber \\
			&=& \max \limits_{L} \{- \alpha(L)+tr(L U^{ \dagger } \rho U) \} \nonumber\\
			&=& \max \limits_{L} \{- \alpha(L)+tr(U L U^{ \dagger } \rho)\}.
		\end{eqnarray}
		
		Since $ L $ is positive semi-definite, we have $  \langle \varphi |L | \varphi \rangle \geq 0 $ and $ \langle \varphi |{L}'| \varphi \rangle= (\langle \varphi | U) L (U^{ \dagger }| \varphi \rangle) \geq0 $ for any $ | \varphi \rangle $, so $ {L}' $ is also positive semi-definite. Furthermore, the unitary transformation does not change the eigenvalues, so $ tr({L}')=tr(L)=a $. As $ U $ is a local unitary matrix, $ | \varphi^{*} \rangle= U^{\dagger } | \varphi \rangle$ is still a separable pure state. In addition, when $ | \varphi \rangle $ takes all separable states, $ | \varphi^{*} \rangle $ could also take all the separable states owing to the fact that $ U^{ \dagger} $ is reversible, then
		\begin{eqnarray}
			\alpha(U L U^{ \dagger})&=& \max \limits_{ | \varphi \rangle \in sep} \langle \varphi | U L U^{\dagger }| \varphi \rangle \nonumber\\
			&=& \max \limits_{ | \varphi \rangle \in sep} ( \langle  \varphi | U) L (U^{ \dagger } | \varphi \rangle) \nonumber\\
			&=& \max \limits_{ | \varphi \rangle^{*} \in sep} \langle \varphi^{*}| L | \varphi^{*} \rangle \nonumber\\
			&=& \max \limits_{ | \varphi \rangle \in sep} \langle \varphi | L | \varphi \rangle \nonumber\\
			&=& \alpha(L).
		\end{eqnarray}
		
		When $ L $ takes all Hermitian matrices, $ ULU^{ \dagger}$ can also take all Hermitian matrices. For $ \forall L' \in \mathcal{H}_{A}\otimes\mathcal{H}_{B} $, we have $ L=U^{ \dagger } L' U \in \mathcal{H}_{A} \otimes \mathcal{H}_{B} $ which satisfies $L'=U L U^{ \dagger }$. Therefore
		\begin{eqnarray}
			C_{w}(U^{ \dagger } \rho U)&=& \max \limits_{L}\{- \alpha(L)+tr(ULU^{ \dagger } \rho) \} \nonumber\\
			&=& \max \limits_{U^{ \dagger } L U} \{- \alpha(U L U^{ \dagger })+tr(U L U^{ \dagger } \rho) \} \nonumber\\
			&=& \max \limits_{L'} \{- \alpha(L')+tr( L' \rho) \} \nonumber\\
			&=&C_{w}( \rho ).
		\end{eqnarray}
	\end{proof}
	
	\begin{Pro}
		$ C_{w}(\rho ) $ is entangled monotone with respect to separable operation $ \varepsilon (\rho ) $, where
		\begin{eqnarray}
			\varepsilon ( \rho )= \sum_{i}^{}(E_{i} \otimes F_{i}) \rho (E_{i} \dagger \otimes F_{i} \dagger ),
		\end{eqnarray}
		and operator completeness is satisfied, which means $ \sum_{i}^{}E_{i}E_{i} \dagger \otimes F_{i}F_{i}\dagger =I_{\mathcal{H}_{A} \otimes \mathcal{H}_{B} } $.
	\end{Pro}
	
	\begin{proof}[\textbf{Proof}]
		We have
		\begin{eqnarray}
			C_{w}(\varepsilon (\rho ) )&=& \max \limits_{L} \{- \alpha(L)+tr( L \varepsilon ( \rho )) \} \nonumber\\
			&=& \max \limits_{L}\{- \alpha(L)+tr( L \sum_{i}^{}(E_{i} \otimes F_{i}) \rho (E_{i} \dagger \otimes F_{i} \dagger )) \} \nonumber\\
			&=&\max\limits_{L} \{- \alpha(L)+ tr( \sum_{i}^{}(E_{i}\dagger \otimes F_{i} \dagger ) L(E_{i}\otimes F_{i})\rho ) \}.
		\end{eqnarray}
		
		Firstly, we prove $\sum_{i}^{}(E_{i} \dagger \otimes F_{i} \dagger ) L(E_{i} \otimes F_{i}))$ conforms to the definition of $ L $, i.e. it satisfies the properties of $ L $ in the definition
		\begin{eqnarray}
			tr(\sum_{i}^{}(E_{i} \dagger \otimes F_{i} \dagger ) L(E_{i} \otimes F_{i}))=tr( \sum_{i}^{}(E_{i}E_{i} \dagger \otimes F_{i}F_{i} \dagger)L)=tr(L)=a>0.
		\end{eqnarray}
		
		$ \sum_{i}^{}(E_{i} \dagger \otimes F_{i} \dagger ) L(E_{i} \otimes F_{i})) $ is obviously a self-conjugate positive semi-definite matrix. So $ \sum_{i}^{}(E_{i} \dagger \otimes F_{i}\dagger ) L(E_{i}\otimes F_{i})) $ meets the definition of $ L $. Next we consider the relationship between $ \alpha(L) $ and $ \alpha(\sum_{i}^{}(E_{i}\dagger \otimes F_{i} \dagger ) L(E_{i} \otimes F_{i})) $:
		\begin{eqnarray}
			\alpha(\sum_{i}^{}(E_{i} \dagger \otimes F_{i}\dagger ) L(E_{i}\otimes F_{i})) &=& \max \limits_{ | \varphi \rangle \in sep} \langle \varphi | \sum_{i}^{}(E_{i} \dagger \otimes F_{i} \dagger ) L(E_{i} \otimes F_{i}) | \varphi \rangle\nonumber\\
			&=&\max\limits_{ | \varphi \rangle \in sep} \sum_{i}^{} \langle \varphi | (E_{i} \dagger \otimes F_{i} \dagger ) L(E_{i} \otimes F_{i}) | \varphi \rangle.
		\end{eqnarray}
		
		Let $ | \varphi^*_{i} \rangle = (E_{i} \otimes F_{i}) | \varphi \rangle/ \sqrt{p_{i}} $, where $ p_{i}=tr( \langle \varphi | (E_{i}\dagger \otimes F_{i}\dagger ) (E_{i}\otimes F_{i}) | \varphi \rangle) $, and $ \sum_{i}^{}p_{i}=1 $. We have
		\begin{eqnarray}
			\alpha(\sum_{i}^{}(E_{i} \dagger \otimes F_{i} \dagger ) L(E_{i} \otimes F_{i})) &=& \max \limits_{ | \varphi \rangle \in sep} \sum_{i}^{} p_{i} \langle \varphi^*_{i} | L | \varphi^*_{i} \rangle \nonumber\\
			&=&\sum_{i}^{} p_{i} \max \limits_{ | \varphi \rangle \in sep} \langle \varphi^*_{i} | L | \varphi^*_{i} \rangle \nonumber\\
			&\leq&\sum_{i}^{} p_{i} \max \limits_{ | \varphi*_{i} \rangle \in sep} \langle \varphi^*_{i} | L | \varphi^*_{i} \rangle \nonumber\\
			&\leq&\alpha(L) .
		\end{eqnarray}
		
		We may clearly arrive at the conclusion that
		\begin{eqnarray}
			C_{w}( \varepsilon ( \rho ) )&=& - \alpha(L)+tr( L \varepsilon ( \rho )) \nonumber\\
			&=& - \alpha(L)+ tr( \sum_{i}^{}(E_{i} \dagger \otimes F_{i} \dagger ) L(E_{i} \otimes F_{i}) \rho ) \nonumber\\
			& \leq& - \alpha( \sum_{i}^{}(E_{i} \dagger \otimes F_{i} \dagger ) L(E_{i} \otimes F_{i}))+ tr( \sum_{i}^{}(E_{i} \dagger \otimes F_{i} \dagger ) L(E_{i} \otimes F_{i}) \rho) \nonumber\\
			&=& - \alpha(L')+ tr( L'\rho	) \nonumber\\
			&=& C_{w}(\rho).
		\end{eqnarray}
	\end{proof}
	
	\begin{Cor}
		$ C_{w}( \rho ) $ is entangled monotone with respect to LOCC (E4).
	\end{Cor}
	\begin{proof}[\textbf{Proof}]
		LOCC belongs to separable operation, and we have proved  $ C_{w}( \rho ) $ is entangled monotone with respect to separable operation.
	\end{proof}
	
	\begin{Pro}
		$ C_{w}(\rho) $ is continuous for $ \rho $, satisfying the property (E5).
	\end{Pro}
	\begin{proof}[\textbf{Proof}]
		Set $ \rho ' = \rho + v_{n}$ and $ \rho ' = \rho -v_{n}$, we could have
		\begin{eqnarray}
			|C_{w}(\rho)-C_{w}(\rho')| \le \max \limits_{L}|tr(Lv_{n})|,
		\end{eqnarray}
		where $ v_{n} $ is a Hermitian matrix. The spectral decompose of $v_{n}$ is $ \sum_{i}^{} \lambda_{i}| \varphi_{i}\rangle\langle\varphi_{i}| $, and $ \sum_{i}^{} \lambda_{i}=0 $. Thus
		\begin{eqnarray}
			tr(Lv_{n})= \sum_{i}^{} \lambda_{i} \langle \varphi_{i}|L| \varphi_{i} \rangle,
		\end{eqnarray}
		where $ \langle \varphi_{i}|L| \varphi_{i} \rangle= \langle \varphi_{i}^{*}| \Lambda | \varphi_{i}^{*}  \rangle \le \lambda_{max}(L) $, $ L=U^{ \dagger} \Lambda U $, $ U $ is a unitary matrix, $ \Lambda $ is a real diagonal matrix with the diagonal elements being the eigenvalues of $ L $, $ | \varphi^{*} \rangle=U| \varphi \rangle $. According to the properties of $ L $, $ |tr(Lv_{n})| $ satisfies
		\begin{eqnarray}
			|tr(Lv_{n})| \le a \sum_{i}^{}| \lambda_{i}|=atr( \sqrt{v_{n}^{ \dagger} v_{n}})=a||v_{n}||_{KF}.
		\end{eqnarray}
		
		When $ \rho'  \rightarrow\rho $, $ v_{n} \rightarrow 0 $, i.e. $||v_{n}||_{KF} \rightarrow 0 $, we have
		\begin{eqnarray}
			0 \le |C_{w}( \rho)-C_{w}( \rho')| \le \max \limits_{L}|tr(Lv_{n})| \le a||v_{n}||_{KF} \rightarrow 0,
		\end{eqnarray}
		which proves that $C_{w}( \rho) $ is continuous.
	\end{proof}
		
	\section{Lower Bound}

	\subsection{Theoretical Derivation}

		We present a lower bound using the idea of finding a specific $ L $ to simplify the calculation. Without loss of generality, we discuss the situation of $ L \in M_{1} $. Then we can set $ L^{'} = aL $ to satisfy $ L^{'} \in M_{a} $. In the following, we discuss the lower bound in two cases:  $ \rho $ is a pure state and $ \rho $ is a mixed state.

	\subsubsection{Pure State}

		We will divide the lower bound of $ C_{w}(\rho) $ into two steps: calculating $ \alpha (L) $ and $ tr(L\rho) $. If we are capable to obtain the estimate of $ \alpha (L) $ in advance, the original problem will be simplified. As a matter of fact, J.Sperling and W.Vogel\cite{J. Sperling2009} have gotten the estimate of $\alpha (L)$. We also demonstrate this result by using Schmidt decomposition rather than the Lagrange multiplier.
	\begin{mythm}
		When $ rank(L)=1 $, we could regard $ L $ as a pure state and we have
		\begin{eqnarray}
			\alpha (L)= \max \limits_{|\varphi \rangle \in sep} \langle \varphi | L | \varphi \rangle =\mu_{1},
		\end{eqnarray}
		where $ \sqrt{ \mu_{1}} $ is the maximum Schmidt coefficient of $ L $.
	\end{mythm}
	\begin{proof}[\textbf{Proof}]
		For $ \forall | \varphi_{pure} \rangle \in \mathcal{H}^{(m)}_{A} \otimes \mathcal{H}^{(n)}_{B}(m \leq n) $, we have
		\begin{eqnarray}
			| \varphi_{pure} \rangle =\sum_{i} \sqrt{ \mu_{i}} | i_{A} i_{B} \rangle,
		\end{eqnarray}
		where $ \sqrt{\mu_{i}} $ ($ i=1,...m $) are the Schmidt coefficients, and $ \mu_{1} \geq \mu_{2} \geq ... \geq \mu_{m} \geq 0 $, $\sum_{i} \mu_{i}=1 $. $ |  i_{A} \rangle $ and $ | i_{B} \rangle $ are the orthonormal basis in $ \mathcal{H}^{(m)}_{A} $ and $ \mathcal{H}^{(n)}_{B} $ respectively. Then $ L= | \varphi_{pure} \rangle \langle \varphi_{pure} | $ could be written as $\sum_{i,j} \sqrt{\mu_{i}\mu_{j}} | i_{A} i_{B} \rangle \langle j_{A} j_{B} | $.

 		Similarly, we could write separable pure state $ | \varphi_{sep} \rangle $ in form of $ \sum_{i,j} a_{i} b_{j} | i_{A}j_{B} \rangle $. By using these representations, we could conclude that
		\begin{eqnarray}
			\langle  \varphi_{sep} | L | \varphi_{sep} \rangle &=& \langle \varphi_{A} | \otimes \langle \varphi_{B} | L | \varphi_{A} \rangle \otimes | \varphi_{B} \rangle \nonumber\\
			&=& \sum_{i,j} \sqrt{ \mu_{i} \mu_{j}}a_{i}b_{i}a_{j}^{*}b_{j}^{*} \nonumber\\
			& \leq&( \sum_{i} \sqrt{ \mu_{i}} | a_{i} | | b_{i} | )^{2} \nonumber\\
			&=&( \sum_{i} \sqrt[4]{ \mu_{i}} | a_{i} | \sqrt[4]{ \mu_{i}} | b_{i} | )^{2}\nonumber\\
			& \leq& ( \sum_{i} \sqrt{ \mu_{i}} | a_{i} |^{2}) ( \sum_{j} \sqrt{ \mu_{j}} | b_{j} |^{2}).
		\end{eqnarray}
		
		The first inequality holds by applying the absolute value inequality and the second one holds due to the Cauchy-Schwarz inequality. Note that $ \sum_{i} | a_{i} |^{2} =1 $ and $ \sum_{j} | b_{j} |^{2} =1 $. We can obtain $ \langle  \varphi_{sep} | L | \varphi_{sep} \rangle \leq \mu_{1} $ and the equal sign works if and only if it meets the condition that $ |a_{1}|=1 $, $ |b_{1}|=1 $, $ a_{i}=0 $, $ b_{i}=0 $ $ (i>1) $. Ultimately, we have $ \alpha (L)= \mu_{1} $.
	\end{proof}

		After getting the estimate of  $ \alpha (L) $, we may substitute it into $ C_{w}(\rho) $. Thus the original problem will be transformed into estimating of $ tr(L\rho) $. In the pure states, we obtain the following result.
	\begin{mythm}
		Let $ \rho $ be a pure  quantum state acting on the space of $ \mathcal{H}_{A}^{(d_{1})} \otimes \mathcal{H}_{B}^{(d_{2})} $. We have
		\begin{eqnarray}
			C_{w}( \rho) \geq \frac{ \| \rho^{T_{A}} \|_{KF}-1}{d},
		\end{eqnarray}
		where $ d=min \{ d_{1}, d_{2} \} $.
	\end{mythm}
	\begin{proof}[\textbf{Proof}]
		It is easy to deduce that $ \rho = \sum_{i,j} \sqrt{ \mu_{i} \mu_{j}} | i_{A} i_{B} \rangle \langle j_{A} j_{B} | $. Then we take a special $ L' $, which has the standard Schmidt form $ L'= \sum_{i,j} \sqrt{m_{i}m_{j}} | i_{A} i_{B} \rangle \langle  j_{A} j_{B} | $, where $ m_{1} \geq m_{2} \geq ... \geq m_{d} \geq 0 $ and $ \sum_{i} m_{i}=1 $. We can notice that $ L $ is a pure state. Then
		\begin{eqnarray}
			C_{w}( \rho) & \geq tr(L' \rho) - \alpha(L')=(\sum_{i} \sqrt{m_{i} \mu_{i}})^{2}-m_{1}.
		\end{eqnarray}
	
		By designating $ m_{i}= \frac{1}{d} $( it meets the above restraint conditions obviously), we can obtain
		\begin{eqnarray}
			C_{w}(\rho) \geq \frac{(\sum_{i} \sqrt{\mu_{i}})^{2}-1}{d}=\frac{ \| \rho^{T_{A}} \|_{KF}-1}{d}.
		\end{eqnarray}
		
		From the Ref. \cite{Kai Chen2005}, we know that $ ( \sum_{i} \sqrt{ \mu_{i}})^{2} $ is equal to $ \| R(\rho) \|_{KF} $ and $ \| \rho^{T_{A}} \|_{KF} $ for a pure state $ \rho $, where $ R(\rho) $ is the realigned matrix of $ \rho $. Here we write $ (\sum_{i} \sqrt{\mu_{i}})^{2} $ as $ \| \rho^{T_{A}} \|_{KF} $. Hence, the inequality in the previous equation holds.
	\end{proof}

		In this way and by means of Schmidt decomposition we succeeded in obtaining the lower bound of $ C_{w}(\rho) $ for pure states. For mixed states, we still attain the expression of $ \alpha (L) $ firstly.

	\subsubsection{Mixed State}

		To solve the problem for the mixed state, we introduce a new form of density operator. An arbitrary state $ \rho $ acting on $ \mathcal{H}_{A}^{(d_{1})} \otimes \mathcal{H}_{B}^{(d_{2})}$ could be Bloch represented as follows
	\begin{eqnarray}
		\rho=\frac{1}{d_{1}d_{2}}  ( I \otimes I+ \sum_{i}{s_{i} \lambda_{i}^{A} \otimes I}+\sum_{j}{I \otimes t_{j} \lambda_{j}^{B} + \sum_{i,j}{r_{ij} \lambda_{i}^{A} \otimes \lambda_{j}^{B}} } ),
	\end{eqnarray}
		where $ I $ stands for identity operator, and $ \{ \lambda_{i}^{A} \}_{i=1}^{d_{1}^{2}-1} $, $ \{ \lambda_{i}^{B} \}_{i=1}^{d_{2}^{2}-1} $ are generalized Gell-Mann matrices, which satisfy $ tr( \lambda_{i} \lambda_{j})=2 \delta_{ij} $, $\lambda_{i}^{\dagger }= \lambda_{i} $. It's clear that $ s_{i} $, $ t_{i} $ and $ r_{ij} $  are all real. Hence, the value of $ \alpha (L) $ could be estimated again when $ L $ is a mixed state.
	\begin{mythm}
		Let $ s_{i}=0 $, $ t_{i}=0 $, $ L $ is written as Bloch representation as follows
		\begin{eqnarray}
			L =\frac{1}{d_{1}d_{2}} ( I \otimes I+\sum_{i,j}{r_{ij }\lambda_{i}^{A} \otimes \lambda_{j}^{B}} ) .
		\end{eqnarray}
	
		Then the estimate of $ \alpha (L)$ is
		\begin{eqnarray}
			\alpha (L)= \frac{1}{d_{1}d_{2}}+ \frac{4}{d_{1}^{2}d_{2}^{2}} \| R_{L} \|_{2},
		\end{eqnarray}
		where $ R_{L}=  ( r_{ij} ) $.
	\end{mythm}
	\begin{proof}[\textbf{Proof}]
		For Hermitian matrix $ L $, $ L $ can be written as
		\begin{eqnarray}
			L = \frac{1}{d_{1}d_{2}}  ( I \otimes I+ \sum_{i}{s_{i} \lambda_{i}^{A} \otimes I}+ \sum_{j}{I \otimes t_{j}  \lambda_{j}^{B} + \sum_{i.j}{r_{ij} \lambda_{i}^{A} \otimes \lambda_{j}^{B}} }).
		\end{eqnarray}
	
		For separable state $\rho$, it can be written as same as $L$:
		\begin{eqnarray}
			\rho= \frac{1}{d_{1}d_{2}} ( I\otimes I+\sum_{i}{s_{i}^{'} \lambda_{i}^{A} \otimes I}+\sum_{j}{I\otimes t_{j}^{'} \lambda_{j}^{B} + \sum_{i,j}{r_{ij}^{'} \lambda_{i}^{A} \otimes \lambda_{j}^{B}} }  ).
		\end{eqnarray}

		Assume $ S $, $ S_{1} $, $ T $, $ T_{1} $ are vectors constituted by $ s_{i} $, $ s_{i}' $, $ t_{j} $, $ t_{j}' $ respectively, $ R_{L}$ and $ R_{1} $ are the correlation matrices of $ L $ and $ \rho $. We could prove $ R_{1}=S_{1}T_{1}^{T} $ as $ \rho $ is a separable state, then
		\begin{eqnarray}
			\alpha (L)=\frac{1}{d_{1}^{2} d_{2}^{2}}(d_{1}d_{2}+ \max \limits_{ \| S_{1} \|_{2}=1, \| T_{1} \|_{2}=1} ( 2 d_{2} S_{1}^{T}S + 2 d_{1} T_{1}^{T}T+4 S_{1}^{T}R_{L}T_{1}) ).
		\end{eqnarray}
	
		Let $ L = \frac{1}{d_{1}d_{2}} ( I \otimes I + \sum{r_{ij} \lambda_{i}^{A} \otimes \lambda_{j}^{B}}) $. We may simplify the above formula as
		\begin{eqnarray}
			\alpha (L)=\frac{1}{d_{1}d_{2}}+\max\limits_{ \| S_{1} \|_{2}=1, \| T_{1} \|_{2}=1} \frac{4}{d_{1}^{2} d_{2}^{2}}S_{1}^{T}R_{L}T_{1} \label{eqa37}.
		\end{eqnarray}
	
		Assume that the singular value decomposition(SVD) form of $ R_{L} $ is $ U\Lambda V^{\dagger} $, $ \Lambda=diag(\sigma_{i})(\sigma_{1}\geq\sigma_{2}\geq...) $. Set $S_{1}^{'T}=S_{1}^{T}U$, $T_{1}^{'}=V^{\dagger}T_{1} $, $ S_{1}^{T}R_{L}T_{1} $ can be written as $ S_{1}^{'T}\Lambda T_{1}^{'} $. We have the following inequality
		\begin{eqnarray}
			S_{1}^{'T}\Lambda T_{1}^{'}&=\sum s_{i}t_{j}\sigma_{i}\leq \sqrt{(\sum|s_{i}|^{2}\sigma_{i})(\sum|t_{i} |^{2}\sigma_{i})}\leq \sigma_{1}. \label{eqa38}
		\end{eqnarray}
	
		Obviously, the upper bound in Eq(\ref{eqa38}) is saturated. From Eq(\ref{eqa37}) and Eq(\ref{eqa38}), we could conclude that
		\begin{eqnarray}
			\alpha (L)&=&\frac{1}{d_{1}d_{2}}+\frac{4}{d_{1}^{2} d_{2}^{2}}\sigma _{max}(R_{L})\\
			&=&\frac{1}{d_{1}d_{2}}+\frac{4}{d_{1}^{2} d_{2}^{2}}  \| R_{L}   \|_{2} .
		\end{eqnarray}
	\end{proof}

		After getting the estimate of  $ \alpha (L) $ when $ L $ is a special mixed state, we may firstly discuss the lower bound of $C_{w}(\rho) $ when $ \rho $ acts on the $ \mathcal{H}^{d_{1}} \otimes \mathcal{H}^{d_{2}} $ space. Then we improve it on the $\mathcal{H}^{2} \otimes \mathcal{H}^{2} $ space.
	
	\begin{mythm}
		Considering the special $ L $ which can be written as the following Bloch representation
		\begin{eqnarray}
			L =\frac{1}{d_{1}d_{2}}  ( I\otimes I+\sum{r_{ij}\lambda_{i}\otimes\lambda_{j}}   ) .
		\end{eqnarray}
	
		Assume $ R_{L}=(r_{ij}) $, when $ R_{\rho} \neq 0 $, we have
		\begin{eqnarray}
			C_{w}(\rho) \geq \frac{2(  \| R_{\rho}   \|_{KF}-1)}{d_{1}^{2}d_{2}^{2} \sqrt{rank(R_{\rho})}}.
		\end{eqnarray}
	
		When $ R_{\rho}=0 $, we have $ C_{w}(\rho) \geq 0 $.
	\end{mythm}
	\begin{proof}[\textbf{Proof}]
		From Theorem 3, the above $ L $ could derive $ \alpha (L)= \frac{1}{d_{1}d_{2}}+\max\limits_{ \| S_{1} \|_{2}=1, \| T_{1} \|_{2}=1} \frac{4}{d_{1}^{2} d_{2}^{2}}S_{1}^{T}R_{L}T_{1} $. Then
		\begin{eqnarray}
			C_{w}( \rho) \geq tr(L \rho)- \alpha(L) = \frac{4}{d_{1}^{2}d_{2}^{2}}(tr(R_{L}^{T}R_{ \rho})- \| R_{L} \|_{2}).
		\end{eqnarray}
	
		Let $ R_{\rho}=U \Sigma V^{T} $ be the singular value decomposition of $R_{ \rho}$, where $ U $ and $ V $ are orthogonal matrices. Assume $ \Sigma' $ is a diagonal matrix which satisfies $tr(\Sigma'^{T} \Sigma)=tr( \Sigma)= \| R_{ \rho}   \|_{KF} $, with the diagonal element to be $ sgn(\sigma_{i}) $, where $ \sigma_{i} $ is the element of $ \Sigma $ in corresponding position. Assume $ R_{L}=\frac{1}{c}U \Sigma' V^{T} $, then we have
		\begin{eqnarray}
			C_{w}(\rho) \geq \frac{4(  \| R_{\rho}   \|_{KF}-1)}{d_{1}^{2}d_{2}^{2}c}.
		\end{eqnarray}
	
		Next we will prove that when $ c=2\sqrt{rank(R_{\rho})} $, $ L\geq 0 $. When $ c=2\sqrt{rank(R_{\rho})} $, and $ rank(R_{\rho}) \neq 0 $, we have
		\begin{eqnarray}
			\| \sum{r_{ij}\lambda_{i} \otimes \lambda_{j}} \|_{HS}&=& \sqrt{tr( \sum{r_{ij}r_{sk} \lambda_{i} \lambda_{s} \otimes \lambda_{j} \lambda_{k}} )} \nonumber\\
			&=&\sqrt{4\sum{r_{ij}^2}}\nonumber\\
			&=&2  \| R_{L}   \|_{HS}
		\end{eqnarray}
		and
		\begin{eqnarray}
			\| R_{L}   \|_{HS}&=&\sqrt{tr(\frac{1}{c^2} V \Sigma'^{T} U^{T} U \Sigma' V)}\nonumber\\
			&=&\frac{1}{c}\sqrt{tr(\Sigma'^{T} \Sigma')}\nonumber\\
			&=&\frac{1}{c}\sqrt{rank(R_{\rho})}.
		\end{eqnarray}
	
		Then $ \| \sum{r_{ij}\lambda_{i}\otimes\lambda_{j}} \|_{HS} = 2 \| R_{L} \|_{HS}=1$. So we have $ \lambda_{max}( \sum{r_{ij} \lambda_{i} \otimes \lambda_{j}} ) \leq \| \sum{r_{ij} \lambda_{i} \otimes \lambda_{j}} \|_{HS} = 1$. Thus $ L = \frac{1}{d_{1}d_{2}} ( I \otimes I + \sum{r_{ij} \lambda_{i} \otimes \lambda_{j}} ) \geq 0$. So we get
		\begin{eqnarray}
			C_{w}(\rho) \geq \frac{2( \| R_{\rho} \|_{KF}-1)}{d_{1}^{2}d_{2}^{2} \sqrt{rank(R_{ \rho})}}.
		\end{eqnarray}
	
		When $ rank(R_{\rho})=0 $, i.e. $ R_{\rho}=0 $, we have $ C_{w}(\rho) \geq 0$, which ends the proof.
	\end{proof}
	
		\begin{mythm}
		For $ \rho \in  \mathcal{H}^{2} \otimes \mathcal{H}^{2} $, we have
		\begin{equation}
			C_{w}(\rho)\geq \frac{1}{4}\max\{ 0, m_{1}+m_{2}-sgn(|R_{\rho}|)m_{3}-1, \frac{m_{1}+m_{2}+sgn(|R_{\rho}|)m_{3}-1}{3} \},
		\end{equation}
		where $ R_{\rho} $ is the correlation matrix of $ \rho $, and $m_{i}$ ($ m_{1} \geq m_{2} \geq m_{3} \geq 0 $) is the singular value of $ R_{\rho} $.
	\end{mythm}
	\begin{proof}[\textbf{Proof}]
		From Theorem 3, when $ L' =\frac{1}{4} ( I\otimes I+\sum{r_{ij}'\lambda_{i}\otimes\lambda_{j}}   )  $, we have $\alpha ( L' )=\frac{1}{4} ( 1+ \| R_{L'} \|_{2} )$, $ R_{L'}=(r_{ij}') $. For $ \rho \in \mathcal{H}^{2} \otimes \mathcal{H}^{2} $, $ \rho =\frac{1}{4} ( I\otimes I +\vec{s} \cdot\vec {\lambda }\otimes I+I\otimes \vec{t}\cdot \vec{\lambda }+\sum_{}^{}r_{ij}\lambda _{i}\otimes \lambda _{j} ) $ and $ R_{\rho}= ( r_{ij} ) $, we can select an appropriate local unitary transformation $ U_{1}\otimes U_{2} $ on $ \rho $ to diagonalize $ R_{\rho} $, then the transformed $ R_{\rho} $ is $ \Sigma_{\rho}= O_{1}^{T}R_{\rho} O_{2}=diag(m_{1},m_{2},sgn(|R_{\rho}|)m_{3}) $ ($ m_{1} \geq m_{2} \geq	m_{3} \geq 0 $), where $O_{1}$ and $O_{2}$ are real orthogonal matrices satisfying $|O_{1}|=|O_{2}|=1$. We can construct $L'=\frac{1}{4}  ( I\otimes I+\sum {r_{ij}}' \lambda _{i}\otimes\lambda_{j} )$, where ${R}'= ( {r_{ij}}' )=O_{1}^{T}\Sigma_{L'}O_{2}$, $ \Sigma_{L'}=diag(r_{1},r_{2},r_{3}) $, then
		\begin{eqnarray}
			C_{w}(\rho) &\geq& \max\limits_{L'\geq 0} \{ tr(L' \rho)-\alpha ( L' ) \}\\
			&=&\max\limits_{L'\geq 0} \{ \frac{1}{4}( tr(R_{L}^{T}R_{\rho})-\max\limits_{i}|r_{i}| )  \} \\
			&=&\max\limits_{L'\geq 0} \{\frac{1}{4}(\textbf{r}\cdot \textbf{m}-\max\limits_{i}|r_{i}|)  \},
		\end{eqnarray}
		where $ \textbf{r}=(r_{1},r_{2},r_{3}) $ and $ \textbf{m}=(m_{1},m_{2},sgn(|R_{\rho}|)m_{3}) $. From Ref. \cite{Ryszard Horodecki1996}, we know that the above matrix $ L' $ is positive semi-definite if and only if $\textbf{r} $ belongs to the tetrahedron $ \Gamma  $ with vertices $ \textbf{t}_{1}=(-1,-1,-1) $, $ \textbf{t}_{2}=(1,1,-1) $, $ \textbf{t}_{3}=(1,-1,1) $, $ \textbf{t}_{4}=(-1,1,1) $. Since tetrahedron $ \Gamma $ is symmetric about three axes and three coordinate planes, $ \textbf{r} $ still belongs to tetrahedron $ \Gamma $ when we change the order of elements or the sign of the two elements.
		
		When $ \textbf{r} $ satisfies $ \frac{1}{4}(\textbf{r}\cdot \textbf{m}-\max\limits_{i}|r_{i}|)= \max\limits_{L'\geq 0} \{\frac{1}{4}(\textbf{r}\cdot \textbf{m}-\max\limits_{i}|r_{i}|)  \}$, we have $r_{1}\geq r_{2}\geq |r_{3}|\geq 0$. Otherwise, For $ \forall \textbf{r}=(r_{1},r_{2},r_{3}) \in \Gamma $, we can change $ \textbf{r} $ to $ \textbf{r}'=(r'_{1},r'_{2},r'_{3}) \in \Gamma (r'_{1}\geq r'_{2} \geq |r'_{3}|)$ and $ \textbf{r}'-\textbf{r}=(c_{1},-c_{1}+c_{2},-c_{2}+c_{3})(c_{i}\geq 0,i=1,2,3) $, then
		\begin{eqnarray}
			\frac{1}{4}(\textbf{r}'\cdot \textbf{m}-\max\limits_{i}|r_{i}|)- \frac{1}{4}(\textbf{r}\cdot \textbf{m}-\max\limits_{i}|r'_{i}|) &=&\frac{1}{4}(\textbf{r}'-\textbf{r})\cdot \textbf{m}\\
			&=&\frac{1}{4}[c_{1}(m_{1}-m_{2})+c_{2}(m_{2}-m_{3})+c_{3}m_{3}]\\
			&\geq& 0.
		\end{eqnarray}
		
		We can simplify the problem to
		\begin{eqnarray}
			C_{w}(\rho)\geq \max\limits_{\textbf{r} \in \Gamma}\frac{1}{4}(\textbf{r} \cdot \textbf{b}),
		\end{eqnarray}
		where $\textbf{b}=(m_{1}-1,m_{2},sgn(|R_{\rho}|)m_{3}) $ and $ r_{1}\geq r_{2} \geq |r_{3}| $. Hence, $ \textbf{r} $ is in a new tetrahedron $ \Gamma' $ with vertices $ \textbf{t}_{1}'=(0,0,0) $, $ \textbf{t}_{2}'=(1,1,-1) $, $ \textbf{t}_{3}'=(1,0,0) $, $ \textbf{t}_{4}'=(1/3,1/3,1/3) $. Since the maximum value of $ \textbf{r}\cdot \textbf{b} $ can be obtained on the corners of the tetrahedron, and $ \textbf{t}_{2}'\cdot \textbf{b} \geq \textbf{t}_{3}' \cdot\textbf{b}$ always holds. Then we have
		\begin{eqnarray}
			C_{w}(\rho)\geq \frac{1}{4}\max\{ 0, m_{1}+m_{2}-sgn(|R_{\rho}|)m_{3}-1, \frac{m_{1}+m_{2}+sgn(|R_{\rho}|)m_{3}-1}{3} \}.
		\end{eqnarray}
	\end{proof}
	
	\subsection{Numerical Simulation }

		The measurable lower bound can be used to check whether the quantum state is entangled. Next we discuss the lower bound for quantum states in systems with different dimensions.

	\subsubsection{$ \rho \in \mathcal{H}^{2} \otimes \mathcal{H}^{2} $}

		When $ \rho \in  \mathcal{H}^{2} \otimes \mathcal{H}^{2} $, we have two measures to discuss the lower bound of the quantum state. Firstly, let's concentrate on this pure quantum state\cite{Ming Li2010}
		\begin{eqnarray}
			\rho = | \varphi \rangle \langle\varphi |,
		\end{eqnarray}
		where $  | \varphi \rangle =(a,0,0,1/\sqrt{2})^{T} /\sqrt{a^{2}+1/2} $. We can get a group of the lower bound of this quantum state by changing the value of $ a $, as shown Fig.1.
	\begin{figure}[H]
	\centering
	\includegraphics[height=6cm,width=8cm]{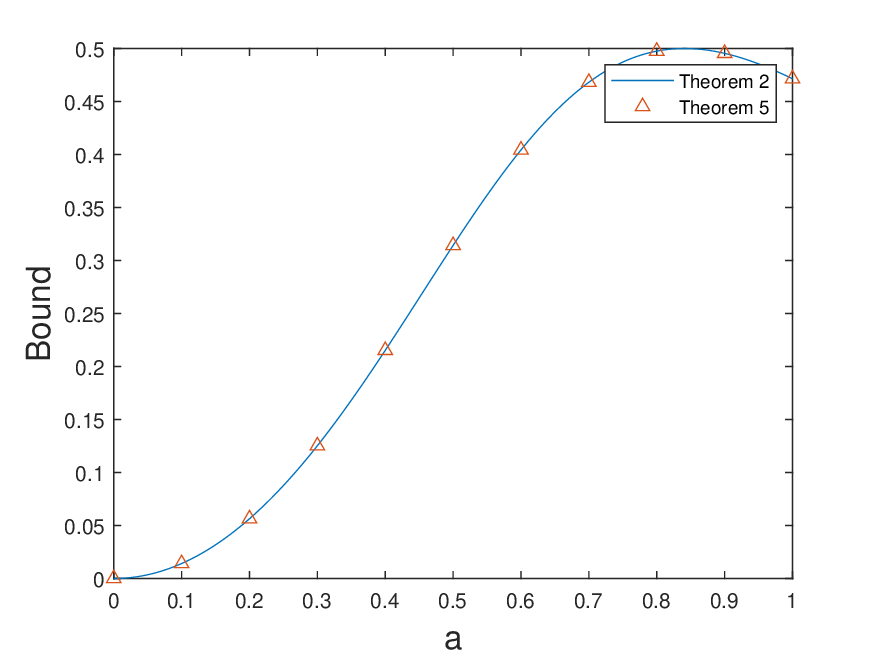}		
	\caption{Lower bounds of $C_{w}(\rho)$, where $ \rho \in  \mathcal{H}^{2} \otimes \mathcal{H}^{2} $ and $\rho$ is pure.}
	\end{figure}
		It is clear that Theorem 2 and Theorem 5 are equivalent in the pure state when $ \rho \in \mathcal{H}^{2} \otimes \mathcal{H}^{2} $. The reason is that the particular $ L $ in Theorem 2 is the maximally entangled state corresponding to the pure state. And Theorem 5 is equivalent to finding the lower bound by checking the maximally entangled state. The difference between them is that Theorem 2 uses the measure of Schmidt decomposition and Theorem 5 uses the measure of Bloch representation. Similar to the pure state, we have two ways to get the value of the lower bound of mixed state. Focusing on the following mixed quantum state:
		\begin{eqnarray}
		\rho =\frac{x}{4}I+ ( 1-x ) | \varphi \rangle \langle\varphi |,
		\end{eqnarray}
		where $| \varphi   \rangle$ has the same definition as the upper one. For $x$=0.1, we have the lower bounds for $C_{w}(\rho)$, as shown Fig.2.
	\begin{figure}[H]
	\centering
	\includegraphics[height=6cm,width=8cm]{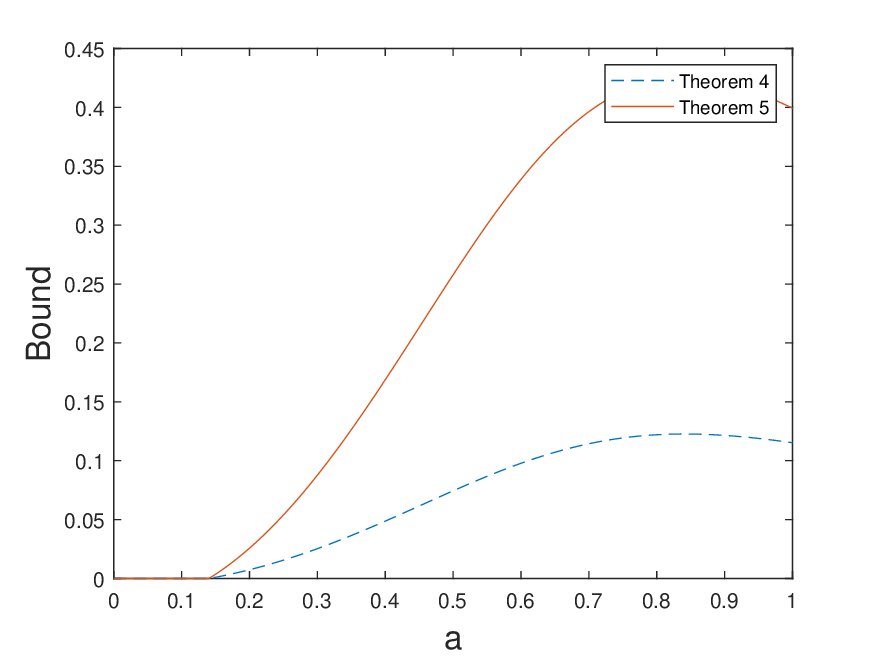}	
	\caption{Lower bounds of $C_{w}(\rho)$, where $ \rho \in \mathcal{H}^{2} \otimes \mathcal{H}^{2} $ and $\rho$ is mixed for $x$=0.1.}
	\end{figure}

		Aiming to compare the different quantum state's influence on the same measure, we figure the similar picture with $x$=0.01, as shown Fig.3.

	\begin{figure}[H]
	\centering
	\includegraphics[height=6cm,width=8cm]{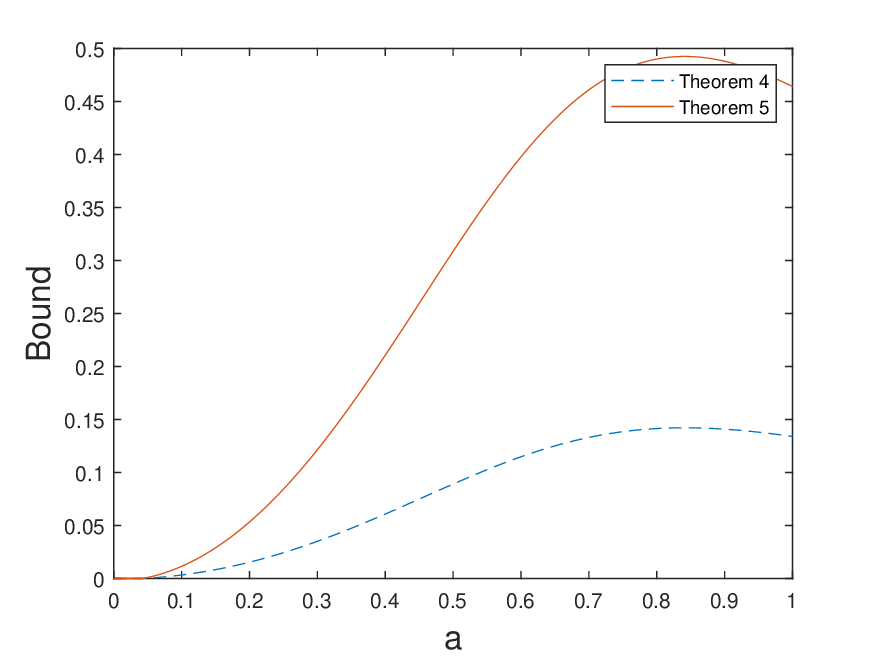}	
	\caption{Lower bounds of $C_{w}(\rho)$, where $ \rho \in \mathcal{H}^{2} \otimes \mathcal{H}^{2} $ and $\rho$ is mixed for x=0.01.}
	\end{figure}

		We can see that the lower bounds of $C_{w}(\rho)$ increase as $x$ gets lower. And Theorem 5 gives us a better result than Theorem 4 when we discuss the quantum state at $ 2\otimes 2 $ Hilbert spaces.

	\subsubsection{$ \rho \in  \mathcal{H}^{3} \otimes \mathcal{H}^{3} $}

		In the same way, we can analyze the quantum state which is higher-dimensional. To avoid the hassle, we'll just cover the quantum states that belong to $\mathcal{H}^{3} \otimes \mathcal{H}^{3} $. Consider the following mixed quantum state.
		\begin{eqnarray}
		\rho =\frac{x}{9}I+ ( 1-x ) | \varphi \rangle \langle\varphi  |,
		\end{eqnarray}
		where $ | \varphi \rangle =(a,0,0,0,1/\sqrt{3},0,0,0,1/\sqrt{3})^{T} /\sqrt{a^2+2/3} $. For $x$=0.01 and $x$=0.1, we can  also get these mixed quantum's lower bound utilizing the Theorem 4, as shown Fig.4.
	\begin{figure}[H]
	\centering
	\includegraphics[height=6cm,width=8cm]{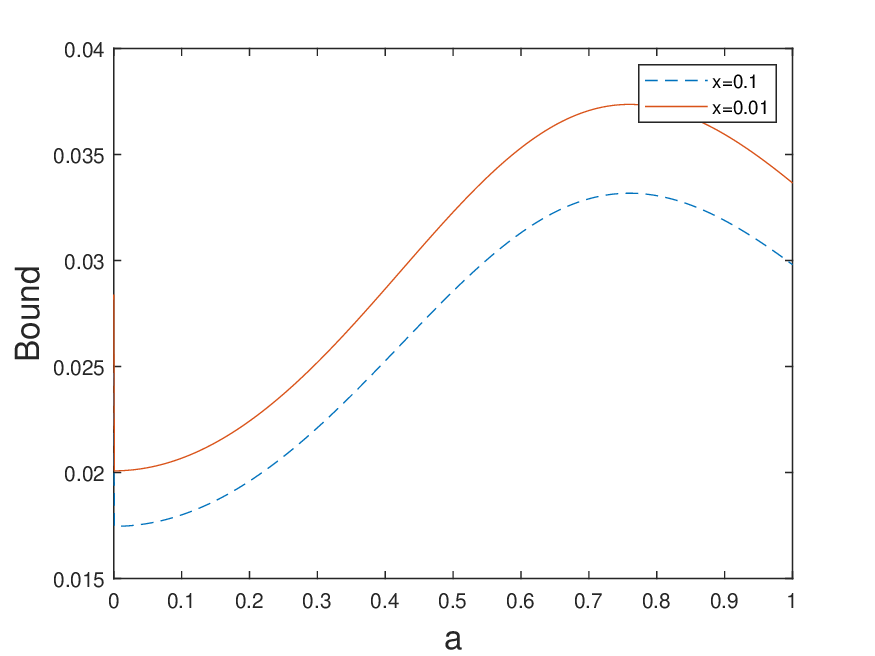}	
	\caption{Lower bounds of $C_{w}(\rho)$, where $ \rho \in \mathcal{H}^{3} \otimes \mathcal{H}^{3} $ and $\rho$ is mixed.}
	\end{figure}

		When we limit it to $tr(L)=1$, we can see that the lower bounds form Theorem 4 are relatively small when we discuss the higher dimensional quantum state. But we can deal with it by limiting the $L \in M_{a}$ and $a$ is higher than 1. If we cut off the unitary part of the quantum state, we can obtain the pure state belonging to $\mathcal{H}^{3} \otimes \mathcal{H}^{3} $, we can get the corresponding lower bound, as shown Fig.5.
	\begin{figure}[H]
	\centering
	\includegraphics[height=6cm,width=8cm]{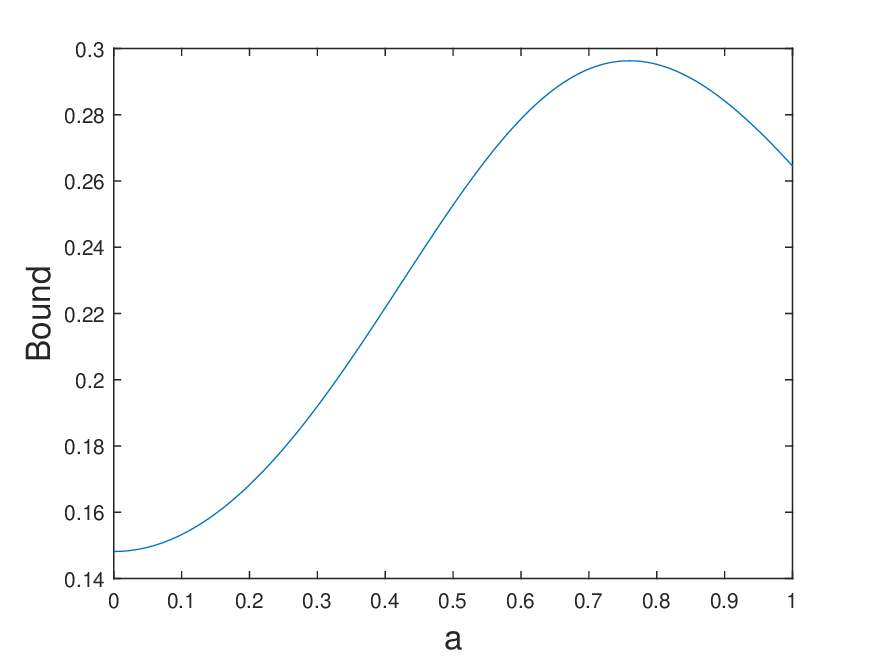}			
	\caption{Lower bounds of $C_{w}(\rho)$, where $ \rho \in \mathcal{H}^{3} \otimes \mathcal{H}^{3} $ and $\rho$ is pure.}
	\end{figure}

	\section{Multipartite Entanglement}

		So far we focused our attention on Hilbert spaces of two finite-dimensional Hilbert spaces. Now we consider the Hilbert space $\mathcal{H}=\mathcal{H}_1 \otimes \cdots \otimes \mathcal{H}_n$. The generalization of the definition of $C_{w}(\rho)$ is
		\begin{eqnarray*}
		C_{w}(\rho) =\max\limits_{L}\{ -tr ( W_{E}\rho) \},
		\end{eqnarray*}
		and it satisfies:
		\begin{eqnarray*}
		&&L \in M_{a},\\
		&&W_{E}=\alpha I-L,\\
		&&\alpha = \max\limits_{ | \varphi \rangle\in sep}\langle  \varphi |L | \varphi \rangle.
		\end{eqnarray*}
	
		It is obvious that there must be a Hermitian matrix $L$ which satisfies $tr  (W_{E}\rho)< 0$, $tr(W_{E}\sigma ) \geq 0$, where $\rho$ is an entangled state, $\sigma$ is a separable state.

	\section{Conclusion}

		In this paper, we define an entanglement measure based on entanglement witness and extend its definition to the multipartite entanglement of quantum states in arbitrary dimensional systems. This measure satisfies some necessary properties of an entanglement measure and we estimate its lower bound on bipartite system. The rationality of the entanglement measure is verified by the estimation of its lower bound by discussing pure and mixed states separately and the corresponding numerical simulation.

	\section*{\bf Acknowledgments} This work is supported by the Shandong Provincial Natural Science Foundation for Quantum Science No.ZR2021LLZ002, and the Fundamental Research Funds for the Central Universities No.22CX03005A.

	\section*{\bf Data availability statement} All data generated or analysed during this study are included in this published article.

\end{document}